\documentclass[letterpaper, 10pt, conference]{ieeeconf}  

\IEEEoverridecommandlockouts                              
\usepackage{threeparttable}
\usepackage{mathtools,dsfont,amsfonts}
\usepackage{microtype}
\usepackage[mathscr]{euscript}

\usepackage[amsmath,thmmarks]{ntheorem}
\allowdisplaybreaks
\DeclareMathOperator*{\argmax}{arg\,max}

\usepackage{xcolor}

\newtheorem{problem}{Problem}

\newtheorem{assumption}{Assumption}

\theorembodyfont{\upshape}
\usepackage{hyperref}
\hypersetup{
    colorlinks=true,
    linkcolor=blue,
    filecolor=magenta,      
    urlcolor=cyan,
    citecolor=Green,
    }
    
\usepackage{booktabs}
\usepackage{graphicx}
\usepackage[sort,compress]{cite}

\newcommand*\obe{\omega}
\newcommand*\budcap{\xi}
\newcommand*\bufcap{\beta}

\usepackage{algorithm}%
\usepackage{algorithmicx}%
\usepackage{algpseudocode}%

\usepackage{soul}
\setstcolor{red}

\begin{document}
\title{Action Recommendations for Sequentially Rational Strategic Agents}
\author{Renyan Sun and Ashutosh Nayyar
\thanks{Renyan Sun and Ashutosh Nayyar are with the Department of Electrical and Computer Engineering, University of Southern California, Los Angeles, CA, USA (email: renyansu@usc.edu, ashutosn@usc.edu). This work was supported in part under {NSF} grants ECCS 2025732 and  IIS 2448885.}%
}
\maketitle
\begin{abstract}
 We consider a finite-horizon discrete-time dynamic system that is jointly controlled by two strategic agents. There is a system designer that has its own reward function but does not have direct control over the agents' actions. We consider an information structure where the current state and all past history are equally accessible by the designer and the agents. The designer sends  action recommendations to the agents at each time step. Each  agent can use the received recommendation and the available  information to choose its action. We are interested in the setting where the designer would like to send recommendations in a way that incentivizes the agents  to adopt  obedient strategies, i.e., to take the action recommended by the designer. Our goal is to find an optimal action recommendation strategy for the designer that maximizes the designer's objective while  ensuring that obedient strategies are \emph{sequentially rational} for the agents. We provide an algorithm for the designer's problem that involves solving a family of linear programs in a backward inductive manner.  
\end{abstract}
\section{Introduction}
Consider a Markov decision process (MDP) jointly controlled by multiple agents. Each agent is strategic (i.e. interested in optimizing its own objective) and can take actions that influence the evolution of the underlying system. A fundamental solution concept in these settings is the Nash equilibrium---a strategy profile (i.e. a strategy for each agent) where no agent profits by unilaterally changing its action \cite{osborne1994course}. However, since agents are self-interested, outcomes arising at Nash equilibrium may not align with the preferences of a system designer or with social welfare objectives. One way a designer can influence agents is through \emph{information design}, i.e., by selectively revealing some information (that only the designer knows) to the agents (see \cite{bergemann2019information} and references therein). However, in environments with complete  information, all agents and the designer have the same information about the system, leaving the designer with no informational advantage. The designer can still achieve some coordination among the agents by sending action recommendations or messages to them.  This is a key idea behind the concept of correlated equilibrium \cite{aumann1987correlated, osborne1994course}, \cite[Chapter 8]{maschler2020game}. 

Much of the foundational literature on correlated equilibrium has focused on \emph{static} games \cite{brandenburger1987rationalizability, hart1989existence, gilboa1989nash, dhillon1996perfect, moreno1996coalition, papadimitriou2008computing}. In \emph{dynamic} games, however, coordinating behavior is considerably more challenging because the state of the system evolves over time and agents need to reassess their future expected payoffs. Since the designer cannot directly control agents' actions, the action recommendation/messaging strategies must be \emph{incentive compatible}\cite{bergemann2019information}. In particular, the actions recommended by the designer must be \emph{sequentially rational}, i.e, at each time and for each possible  realization of the the agents' information, no agent can benefit from unilaterally deviating from the  designer's recommended  action (see Definition~\ref{def:sequential rationality} for details). This requirement is stronger than standard Nash equilibrium-based incentive compatibility which does not account for the sequential nature of the problem.
In the game theory literature, coordinating strategic behavior under these strict dynamic constraints is formalized as an extensive-form correlated equilibrium \cite{myerson1986multistage, forges1986approach, von2008extensive, huang2008computing, farina2019correlation}. While the algorithmic literature for computing these equilibria relies heavily on finite tree-form representations, we tackle this problem directly within an MDP-based dynamic game (i.e. Markov game) setting.

In our MDP-based system model, the designer's problem can be described as follows: the designer would like to find an action recommendation strategy that maximizes the designer's total expected reward over a finite time horizon subject to agents' sequential rationality constraints.  Conventionally, MDPs with additional operational  requirements are modeled as Constrained MDPs (CMDPs)\cite{altman2021constrained, borkar2005actor}. In standard CMDPs, a constraint is of the form that the total expected value of a constraint reward function should exceed a certain threshold. However, this formulation is not suitable for capturing sequential rationality constraints that need to hold at \emph{each time instant} and for  \emph{each information realization}. Such constraints are stricted than the CMDP constraints that only look at agents' total expected reward over the entire time horizon.


In this paper, we consider a dynamic setting with one designer and two agents (agent $1$ and agent $2$) operating under complete symmetric information where the current state and all past history are fully accessible to everyone. At each time step, the designer sends a recommended action pair/message pair to the agents. The agents use these recommendations and the available information to select their actions. The agents' actions influence the evolution of the underlying dynamic system and determine the agents' and the designer's rewards. We investigate a scenario where the designer aims to incentivize the agents  to adopt obedient strategies, i.e, to take the action recommended by the designer. Our goal is to find an optimal action recommendation/messaging strategy for the designer that maximizes the designer's objective while  ensuring that obedient strategies are \emph{sequentially rational} for the agents.


Our approach proceeds as follows. We first restrict our attention to Markovian messaging strategies. For a fixed Markovian strategy, we analyze sequential rationality to establish \emph{necessary and sufficient} conditions under which obedient strategies are sequentially rational.
This leads to a reformulation of the designer's strategy optimization problem. We develop an algorithm to solve this optimization problem by solving a family of linear programs in a backward inductive manner. Our methodology is similar to the one used   in \cite{sun2025optimal, sun2025optimalarXiv} for an information design problem.  Finally, we argue that the restriction to Markovian strategies for the designer is without loss of optimality. 


\emph{Notation:}
Random variables are denoted by upper case letters; their realizations by corresponding lower case letters. All random variables take values from finite sets which are denoted in the calligraphic font of the corresponding upper case letter. For time indices $t_1\leq t_2$, ${X}_{t_1:t_2}$ is a short hand notation for the variables $({X}_{t_1},{X}_{t_1+1},...,{X}_{t_2})$. Similarly, ${X}^{0:2}$ is the short hand notation for the collection of variables $({X}^0, {X}^1, {X}^2)$. $\mathbb{P}(\cdot)$ denotes probability; $\mathbb{E}[\cdot]$ denotes expectation. $\mathbb{P}^g(\cdot)$ ($\mathbb{E}^g[\cdot]$) denotes that the probability (resp. expectation) depends on the choice of  function(s)

\section{Model and Problem Formulation}\label{sec:model1}
We consider a discrete-time dynamic system that is jointly controlled by two agents\footnote{Our approach naturally extends to more than 2 agents.}. For each time $t\in\{1,2,\ldots,T\}$, $X_t \in \mathcal{X}_t$ is the state of the system at time $t$,  $U_t^i \in \mathcal{U}^i_t $ is the agent $i$'s action  at time $t$. The system evolves as follows
\begin{equation}\label{eq:dyn}
    X_{t+1} = f_t(X_t, U_t^1, U_t^2, N_t),
\end{equation}
where $N_t \in \mathcal{N}_t$ is the noise in the dynamic system at time $t$. The agents are strategic and interested in optimizing their own total expected rewards (which will be specified later in this section).

In addition to the two agents, there is a designer present who is interested in coordinating and influencing the agents' behavior. At each time $t$, a designer sends a 
recommended action (also referred to as a message) $M^i_t \in \mathcal{U}^i_t$ to  agent $i$. The agent is free to either follow this recommended action or take a different action. We describe the designer's and the agents' decision-making in more detail below. 

\emph{Information Structure and Strategies:} The designer, agent 1 and agent 2 at time $t$ have access to the same information $C_t$ which comprises the current state of the system and the history of all past states, actions and messages, i.e., 
\begin{equation}\label{eq:info structure}
    C_t = \{X_{1:t}, U_{1:t-1}^{1,2}, M_{1:t-1}^{1,2}\}.
\end{equation}
 Let $\mathcal{C}_t$ denote the space of all possible realizations of $C_t$.

At time $t$, the designer determines a \emph{recommended action pair/message pair} $(M_t^1, M_t^2) \in \mathcal{U}_t^1 \times \mathcal{U}_t^2$ by sampling from a probability distribution $D^m_t$ on $\mathcal U_t^1 \times \mathcal U_t^2$. This distribution is chosen by the designer as a function of $C_t$, i.e.,
\begin{equation}\label{eq:messaging strategy}
  (M^1_t,M^2_t) \sim D_t^m, \text{~~and~~} D_t^m = g_t^m(C_t).
\end{equation}
The function $g^m_t$ is referred to as the \emph{designer's messaging strategy at time $t$}.
 We call the collection $g^m:=(g^m_1,g^m_2,\ldots,g^m_T)$ the designer's messaging strategy. Let $\mathcal{G}^m$ denote the set of all possible messaging strategies for the designer. For notational convenience, $g^m_t(m^1_t, m^2_t|c_t)$ denotes the \emph{probability} of generating the message pair $m^{1}_t,m^2_t$ given the realization of the information $c_t$ under the strategy $g^m_t$ at time $t$.

 The designer sends message $M^i_t$ to agent $i$. After agent $i$ receives the message $M_t^i$ from the designer, it generates an action as a function of its information and the message it received at time $t$, i.e., 
\begin{equation}\label{eq:agent_g}
    U_t^i = g_t^i(M_t^i, C_t),
\end{equation}
where  $g_t^i$ is the agent $i$'s action strategy  at time $t$. The collection $g^i := (g^i_1,g^i_2,...,g^i_T)$ is referred to as the \emph{agent $i$'s action strategy}. Let $\mathcal{G}^i$ denote the set of all possible action strategies for agent $i$. The strategy triplet  $g:=(g^m, g^1, g^2)$, is called the  \emph{strategy profile}. $(g^m, g^1, g^2)_{t:T}$ denotes the strategies used from time $t$ to $T$. 

\emph{Obedient strategies:} A particular choice of agent strategies that we will be interested in are the so-called obedient strategies. These are strategies where the agent simply takes the action recommended by the designer. We denote an obedient strategy for agent $i$ be $\obe^i$, i.e,
\begin{equation}\label{eq:obe_1}
    U_t^i = \obe_t^{i}(M_t^i, C_t) = M^i_t.
\end{equation}

{We assume that  the initial state $X_1$ and  the noise variables $N_t,  t=1,2,\ldots,T,$ are finite-valued, mutually independent random variables with the distribution of $X_1$ being $P_{X_1}$ and the distribution of $N_t$ being $Q_t$. Further, all system variables (i.e., states, actions, messages, information, etc.) take values in finite sets.}

\emph{Reward structure:} At each time $t$, agent $i$ receives a reward $r^i_t(X_t, U_t^1,U_t^2)$ and the designer receives a reward $r^0_t(X_t, U_t^1,U_t^2)$. The total expected reward for agent $i$ under the strategy profile $g:= (g^m, g^1, g^2)$ is given as:
\begin{equation}\label{eq:agent reward-to-go}
    J^i(g^m, g^1, g^2) := \mathbb{E}^{g}\left[\sum_{t=1}^T r^i_t(X_t, U_t^1,U_t^2)\right].
\end{equation}
Similarly, the designer's total expected reward under the strategy profile $g = (g^m, g^0, g^1)$ is given  as:
\begin{equation}\label{eq:designer reward-to-go}
    J^0(g^m, g^1, g^2) := \mathbb{E}^{g}\left[\sum_{t=1}^T r^0_t(X_t, U_t^1,U_t^2)\right].
\end{equation}

\emph{Incentive compatibility for the agent (i.e., Sequential Rationality of Obedient Strategies):} A minimal requirement for each agent $i$ to 
adopt the obedient strategy $\obe^i$ is that it yields a total expected reward no less than that of any other strategy. We will adopt a stronger notion of incentive compatibility based on sequential rationality \cite[Chapter 12]{osborne1994course}. This   requires that  $\obe^i$ is optimal for agent $i$ at every time step and for every information realization. We formalize this in the definition below.

\begin{definition}\label{def:sequential rationality}
   We say that the obedient strategies $\obe^1, \obe^2$ satisfy sequential rationality (SR) with respect to the designer messaging strategy $g^m$  if the following is true\footnote{We use $-i$ to indicate all agents except agent $i$ or the designer.}:\\
    {For each agent $i$, each time $t$ and each possible realization $c_t = (x_{1:t}, u^{1,2}_{1:t-1}, m^{1,2}_{1:t-1})$ of information at time $t$,}
    \begin{align}
    &\mathbb{E}^{(g^m, \obe^i, \obe^{-i})_{t:T}}\left[\sum_{k=t}^T r^i_k(X_k, U_k^{1,2}) \Bigm| c_t\right] \notag \\
    &\geq \mathbb{E}^{(g^m, g^i, \obe^{-i})_{t:T}}\left[\sum_{k=t}^T r^i_k(X_k, U_k^{1,2}) \Bigm| c_t\right]~\forall g^i \in  \mathcal{G}^i. \label{eq:sr}
     \end{align}
\end{definition}

The expectation on the left hand side of \eqref{eq:sr} is to be interpreted as follows: Given $c_t$, and with future states, actions, messages and information variables generated using strategies $(g^m, \obe^i, \obe^{-i})_{t:T}$, the  left hand side of \eqref{eq:sr} is the expected reward-to-go for  agent $i$. A similar interpretation holds for the right hand side of \eqref{eq:sr}.     
If $\obe^1,\obe^2$ satisfy sequential rationality (SR) with respect to the designer messaging strategy $g^m$, we will denote this by $(\obe^1,\obe^2)$ satisfies $SR(g^m)$.

We state the designer's strategy optimization problem below.
\begin{problem}\label{prob:no comp}
 The designer's goal is to find an optimal messaging strategy $g^m$ that maximizes the designer's total expected reward while ensuring  that the \emph{obedient strategies} $(\obe^1, \obe^2)$ satisfy sequential rationality with respect to $g^m$ as per Definition \ref{def:sequential rationality}. That is, 
    the designer would like to solve the following strategy optimization problem:
 \begin{align*}
    &\max_{g^m \in \mathcal{G}^m} J^0(g^m, \obe^1, \obe^2) \\
    &\hspace{10pt}\text{s.t. $(\obe^1,\obe^2)$ satisfies $SR(g^m)$.}
    \end{align*}
\end{problem}

 \section{Restriction to Markovian Designer Strategies}\label{sec:SA-Markovian}
 One challenge in solving Problem \ref{prob:no comp} is that the space  of possible designer strategies at time $t$ increases doubly exponentially with time. This is because the number of possible realizations of $c_t$ increases exponentially with time and the strategy $g^m_t$ needs to specify a probability distribution on message pairs for each realization of $c_t$. We therefore first consider the case where the designer can only use \emph{Markovian} strategies, i.e., strategies where $g^m_t$ is a function of the current state $X_t$ only. This restriction is motivated by the well-known result that Markovian strategies are  optimal for Markov Decision Processes (MDP) \cite{kumar2015stochastic}. Moreover, we will show  in Section~\ref{sec:SA-General} that restriction to Markovian strategies for the designer is without loss of optimality.

 Specifically, we assume in this section that the designer operates as follows: at time $t$, the designer determines a \emph{recommended action pair/message pair} $(M_t^1, M_t^2) \in \mathcal{U}_t^1 \times \mathcal{U}_t^2$ by sampling from a probability distribution $D^m_t$ on $\mathcal U_t^1 \times \mathcal U_t^2$ which is chosen by the designer as a function of $X_t$, i.e.,
\begin{equation}
  (M^1_t,M^2_t) \sim D_t^m, \text{~~and~~} D_t^m = g_t^m(X_t).
\end{equation}  We denote the set of Markovian messaging strategies by $\mathcal{G}^m_{Markov}$.
$g^m_t(m^1_t, m^2_t|x_t)$ denotes the probability of generating the messages $m^{1,2}_t$ given the realization of the state   $x_t$ under the strategy $g^m_t$ at time $t$. Further, we denote by $g_t^m(m_t^i | x_t)$ the probability of message $m^i_t$ given the realization of the state   $x_t$ under the strategy $g^m_t$. \big($g^m_t(m_t^i | x_t)$ can be obtained by summing $g^m_t(m^i_t,m^{-i}_t|x_t)$ over all possible $m_t^{-i}$ \big). We now have the following problem formulation with Markovian designer strategies.

\begin{problem}\label{prob:no comp markov}
 The designer's goal is to find an optimal Markovian messaging strategy $g^m$ that maximizes the designer's total expected reward while ensuring  that the \emph{obedient strategies} $(\obe^1, \obe^2)$ satisfies sequential rationality with respect to $g^m$ as per Definition \ref{def:sequential rationality}. That is, 
    the designer would like to solve the following strategy optimization problem:
 \begin{align*}
    &\max_{g^m \in \mathcal{G}^m_{Markov}} J^0(g^m, \obe^1, \obe^2) \\
    &\hspace{10pt}\text{s.t. $(\obe^1,\obe^2)$ satisfies $SR(g^m)$.}
    \end{align*}
\end{problem}

Since the $SR$ conditions of \eqref{eq:sr} must hold at each time and for each possible realization of information, our solution approach begins by characterizing $SR$ conditions through backward induction. This will allow us to decompose the optimization in Problem \ref{prob:no comp markov} in a sequential manner.

\subsection{Reformulation of the Constraint in Problem \ref{prob:no comp markov}}\label{subsec:constraints markov}
Let $g^m$ be an arbitrarily fixed Markovian designer  strategy. In this section, we reformulate the condition $(\obe^1, \obe^2)$ satisfies $SR(g^m)$ using linear inequalities involving $g^m$. To achieve this, we first recursively define the  following value functions for agent $i$ $(i = 1,2)$:
\begin{equation}
    W_{T+1}^i(x_{T+1}) := 0, \label{eq:W_defa}
\end{equation}
and for $t \le T,$
\begin{align}
    W_t^i(x_t) &:= \mathbb{E}[r_t^i(x_t, M_t^1, M_t^2)\notag \\
    &\hspace{20pt}+ W_{t+1}^i(f_t(x_t, M_t^1, M_t^2, N_t))|X_t=x_t].
    \label{eq:W_def}
\end{align}
The expectation in \eqref{eq:W_def} is with respect to the probability distribution of $M_t^1, M_t^2$ under $g^m_t$ and the probability distribution of   $N_t$ conditioned on $X_t = x_t$. More explicitly, $W_t^i(x_t)$ can be written as the following expression:
\begin{align}\label{eq:agent val func}
    &W_t^i(x_t) = \sum_{\substack{m^1 \in \mathcal U_t^1\\ m^2 \in \mathcal U_t^2, n \in \mathcal{N}_t}}
Q_t(n)g^m_t(m^1, m^2 | x_t) \notag \\
&\times\left[r_t^i(x_t, m^1, m^2)+ W_{t+1}^i(f_t(x_t, m^1, m^2, n))\right],
\end{align}
Using backward induction, one can verify that $W_t^i(x_t)$ is the left hand side of \eqref{eq:sr} in Definition \ref{def:sequential rationality} for all $c_t = (x_t,\cdot)$ 

Consider a message $m^i_t$ such that $g^m_t(m^i_t|x_t) >0$. We define the following probability distribution on $M_t^{-i}, N_t$ given $m^i_t$ and $x_t$:
\begin{equation}\label{eq:mu}
    \mu_t^i(m_t^{-i}, n_t| m_t^i, x_t) =\frac{Q_t(n_t)g^m_t(m_t^1, m_t^2|x_t)}{g^m_t(m_t^i|x_t)}.
    \end{equation}

We first present a lemma that establishes necessary and sufficient conditions for $(\obe^1,\obe^2)$ to satisfy $SR(g^m)$.
\begin{lemma}\label{lemma:sr equiv no comp}
    $(\obe^1,\obe^2)$ satisfies $SR(g^m)$ if and only if the following statement is true for each $i=1,2, j \neq i$,   for  each $t$, and  for all $x_t \in \mathcal{X}_t, m_t^i \in \mathcal U_t^i$ such that $g_t^m(m_t^i | x_t) > 0$:
    \begin{align}
         &m_t^i \in \argmax_{u \in \mathcal{U}_t^i} \mathbb{E}^{\mu^i_t(\cdot|m^i_t,x_t)}[r_t^i(x_t, M_t^{j}, u)  \notag \\
            &\hspace{1.2in}+W^i_{t+1}(f_t(x_t, M_t^{j}, u, N_t))], \label{eq:sr equiv no comp}
        \end{align}  
        where the expectation is with respect to the probability distribution defined in \eqref{eq:mu}.
\end{lemma}
\begin{proof}
    See Appendix \ref{app:sr equiv no comp}.
\end{proof}

Building upon Lemmas \ref{lemma:sr equiv no comp}'s characterization of $(\obe^1,\obe^2)$ satisfying $SR(g^m)$, we now show its equivalence to inequalities that are linear in $g^m_t$. To do so, we first expand \eqref{eq:sr equiv no comp} into the following collection of inequalities:
\begin{align}
&\mathbb{E}^{\mu_t^i(\cdot|m^i_t,x_t)}[r_t^i(x_t, M_t^{j}, m_t^i) +  W^i_{t+1}(f_t(x_t, M_t^j, m_t^i, N_t))] \notag \\
&\geq \mathbb{E}^{\mu_t^i(\cdot|m^i_t,x_t)}[r_t^i(x_t, M_t^{j}, u) +  W^i_{t+1}(f_t(x_t, M_t^j, u, N_t))] \notag\\
&\forall {u}\in \mathcal{U}^i_t. \label{eq:lp1}
    \end{align}
    
Consider the right hand side of \eqref{eq:lp1}. We can evaluate this expectation as follows:
\begin{align}
    &\sum_{m^j \in \mathcal U_t^j, n \in \mathcal{N}_t}
    \mu_t^i(m^j, n|m_t^i, x_t)\notag \\
    &\hspace{60pt}\times\Big[r_t^i(x_t, m^j, u)+ W^i_{t+1}(f_t(x_t,m^j, u, n))\Big], \notag \\
    & =\sum_{m^j \in \mathcal U_t^j, n \in \mathcal{N}_t}
    \frac{Q_t(n)g^m_t(m^j, m_t^i|x_t)}{g^m_t(m_t^i|x_t)}\notag\\
    &\hspace{30pt}\times\Big[r_t^i(x_t, m^j, u)+ W^i_{t+1}(f_t(x_t,m^j, u, n))\Big], \label{eq:lp2}
    \end{align}
    where we have used the definition of $\mu_t^i$ from \eqref{eq:mu}. (Recall that $g^m_t(m_t^i | x_t) >0$ in \eqref{eq:sr equiv no comp}). Writing a similar expression for the left hand side of \eqref{eq:lp1} and canceling $g^m_t(m_t^i | x_t)$ results in the following set of inequalities that are linear in $g_t^m$:
    \begin{align}
    &\sum_{m^j \in \mathcal U_t^j, n \in \mathcal{N}_t}
    Q_t(n)g^m_t(m^j, m_t^i|x_t)\notag\\
    &\hspace{30pt}\times\Big[r_t^i(x_t, m^j, m_t^i)+ W^i_{t+1}(f_t(x_t,m^j, m_t^i, n))\Big]\notag\\
    &\geq \sum_{m^j \in \mathcal U_t^j, n \in \mathcal{N}_t}
    Q_t(n)g^m_t(m^j, m_t^i|x_t)\notag\\
    &\hspace{30pt}\times\Big[r_t^i(x_t, m^j, u)+ W^i_{t+1}(f_t(x_t,m^j, u, n))\Big], \notag \\
    & \hspace{2cm}\forall {u}\in \mathcal{U}^i_t. \label{eq:lp3}
\end{align}
Thus, the condition in Lemma \ref{lemma:sr equiv no comp} requires \eqref{eq:lp3} to hold for each $i$ and for all $m_t^i, x_t$ where $g^m_t(m_t^i|x_t) > 0$. Further, if for some  $m_t^i, x_t$ $g^m_t(m_t^i|x_t) = 0,$  then it follows that $g^m_t(m_t^j, m^i_t| x_t) =0$ for all $m_t^j$ and hence \eqref{eq:lp3} is trivially true in this case since  both sides equal 0. We summarize the above discussion in the following lemma.
\begin{lemma}\label{lemma:sr equiv no comp linear}
   Consider an arbitrary Markovian designer strategy $g^m$ and define the value functions $W^i_{T+1}, \ldots, W^i_1,\:i=1,2$ using \eqref{eq:W_defa} and \eqref{eq:agent val func}. Then,  $(\obe^1, \obe^2)$ satisfies $SR(g^m)$ if and only if the inequalities in \eqref{eq:lp3} hold for  $t=1,\ldots T, i=1,2, j \neq i,$ and for all  $x_t \in \mathcal{X}_t, m_t^i \in \mathcal U_t^i$.
\end{lemma}

\subsection{Decomposition of Problem \ref{prob:no comp markov} into Nested Linear Programs}\label{sec:decomposition markov} 
Based on Lemma \ref{lemma:sr equiv no comp linear}, Problem \ref{prob:no comp markov} can be equivalently expressed as follows: The designer would like to find a Markovian strategy $g^m$, and the value functions $W^i_{T+1}, \ldots, W^i_1$ (defined in \eqref{eq:agent val func} and \eqref{eq:W_defa}) such that the constraints in \eqref{eq:lp3} hold and the designer's total expected reward is maximized. This yields the following reformulation of Problem \ref{prob:no comp markov}:
\begin{align*}
   &\textbf{Global Problem:} \max_{g^m, W^1_{1:T}, W^2_{1:T}} J^0(g^m, \obe^1, \obe^2) \\
    &\text{s.t.} \text{~for $t=T,T-1,\ldots,1,$ $i=1,2,$} 
    \\  
    &\hspace{18pt}\text{the inequalities in \eqref{eq:lp3} hold for all $m_t^i, x_t,$}\\
    &\hspace{18pt}W_t^{i}(x_t) \text{ satisfies \eqref{eq:agent val func} for all $x_t$ (with ~$W_{T+1}^{i} (\cdot) =0$)}.
\end{align*}

We refer to this formulation as the \emph{Global Problem} since both its objective and constraints span the entire time horizon. Solving it directly can be computationally demanding due to the sheer number of optimization variables and the presence of non-linear constraints (e.g., the products of $g_t^m(\cdot | x_t)$ and $W_{t+1}^i(\cdot)$ in \eqref{eq:lp3}). In this section, we develop a sequential decomposition of the Global Problem into a set of smaller, tractable subproblems.

Given that the constraints of the Global Problem are backward inductive with respect to the value functions $W_T^i,\ldots,W_1^i$, we seek a corresponding backward-inductive decomposition of its objective.
To this end, we introduce the following value functions for the designer, which are analogous to agent $i$'s value functions $W_t^i$ defined earlier:
\begin{equation}
    V_{T+1}(x_{T+1}) := 0, \label{eq:V_defa}
\end{equation}
and for $t \le T,$
\begin{align}
    V_t(x_t) &:= \mathbb{E}[r_t^0(x_t, M_t^1, M_t^2)\notag \\
    &\hspace{20pt}+ V_{t+1}(f_t(x_t, M_t^1, M_t^2, N_t))|X_t=x_t], \label{eq:V_def designer}
\end{align}
More explicitly, $V_t(x_t)$ can be written as the following expression:
\begin{align}\label{eq:designer val func}
    &V_t(x_t) = \sum_{\substack{m^1 \in \mathcal U_t^1\\ m^2 \in \mathcal U_t^2, n \in \mathcal{N}_t}}
Q_t(n)g^m_t(m^1, m^2 | x_t) \notag \\
&\times\left[r_t^0(x_t, m^1, m^2)+ V_{t+1}(f_t(x_t, m^1, m^2, n))\right].
\end{align}
We now construct a backward inductive sequence of optimization problems using the functions $V_{T+1},\ldots, V_1$. We start at time $T$. For  a realization $x_T$ of $X_T$, we formulate the following optimization problem:
\begin{align*}\label{equ: nested problem objective}
&\mathbf{LP}_T(x_T):\hspace{10pt}\max_{\substack{g_T^m(\cdot|x_T), \\V_T(x_T), W_T^1(x_T), W_T^2(x_T)}} V_T(x_T)\notag \\
&\text{s.t.} \text{~for $t=T,$}
\\
&\hspace{25pt}\text{the inequalities in \eqref{eq:lp3} hold  for all $m_T^i,\:i=1,2$} \notag \\
    &\hspace{25pt}W_T^{1,2}(x_T) \text{ satisfy }\eqref{eq:agent val func}\text{~(with $W^{1,2}_{T+1}(\cdot) =0$),} \notag\\
    &\hspace{25pt}V_T(x_T) \text{ satisfies }\eqref{eq:designer val func} \text{~(with $V_{T+1}(\cdot) =0$)}.
\end{align*}
Since the objective and constraints of the preceding optimization problem are linear in its variables $g_T^m(\cdot|x_T), V_T(x_T), W^1_T(x_T), W^2_T(x_T)$, we refer to  it as the linear program $\mathbf{LP}_T(x_T)$. Solving the family of linear programs $\mathbf{LP}_T(x_T)$ for each $x_T \in \mathcal{X}_T$ yields the functions $V_T(\cdot)$ and $W^{1,2}_T(\cdot)$, which are subsequently used to formulate a linear program at time $T-1$ for each $x_{T-1} \in \mathcal{X}_{T-1}$. This procedure in then repeated inductively for earlier time steps. Algorithm~\ref{alg:1} summarizes this backward induction procedure for $t=T, T-1, \ldots,1$, where solving the linear program  $\mathbf{LP}_t(x_t)$ yields (among other things)  $g^m_t(\cdot|x_t)$.
\begin{algorithm}[H]
\caption{}\label{alg:1}
\begin{algorithmic}
    \State{$W^1_{T+1}(\cdot) = W^2_{T+1}(\cdot) = V_{T+1}(\cdot) =0$}
    \For{$t=T,T-1,\ldots,1$}
        \For{ each $x_t \in \mathcal{X}_t$} 
            \State{
            $g_t^m(\cdot|x_t), V_t(x_t), W_t^1(x_t), W_t^2(x_t) =$ 
            \par \hspace{15pt}Solution of the following  linear program $\mathbf{LP}_t(x_t)$
            \vspace{10pt}
            \par
                \hspace{10pt}$\mathbf{LP}_t(x_t) :\max_{\substack{g_t^m(\cdot|x_t),\\V_t(x_t), W^1_t(x_t),W^2_t(x_t)}} V_t(x_t)$ \par
                \hspace{10pt}s.t. the inequalities  in \eqref{eq:lp3} holds $\forall \:m_t^i,$ and $i=1,2$,
                \par
                \hspace{25pt}$W_t^i(x_t)$ satisfy \eqref{eq:agent val func} for each $i=1,2$, 
                \par
                \hspace{25pt}$V_t(x_t)$ satisfies \eqref{eq:designer val func}}.
            \EndFor
            \State{$g^m_t = \{g_t^m(\cdot|x_t)\}_{x_t \in \mathcal{X}_t}$}
    \EndFor
     \State{\Return{$g^m = (g_1^m, \ldots, g_T^m)$}}
\end{algorithmic}
\end{algorithm}

\begin{theorem}\label{thm:alg1}
   The designer messaging strategy $g^m$ returned by Algorithm \ref{alg:1} is an optimal solution for Problem \ref{prob:no comp markov}.
\end{theorem}
\begin{proof}
    See Appendix~\ref{appendix: proofofalgorithm}.
\end{proof}
\begin{remark}
    While our analysis thus far has assumed a fixed action space $\mathcal{U}^i_t$ for agent $i$ at time $t$, the results extend naturally to state-dependent action spaces where $\mathcal{U}^i_t(x_t)$ is agent $i$'s action space   at time $t$ if $X_t=x_t$.
\end{remark}

\section{Extension to General Messaging Strategies}\label{sec:SA-General} 
In Section~\ref{sec:SA-Markovian}, we restricted the designer's messaging strategy to depend only on the current state $X_t$, yielding a Markovian formulation. We now return to the original formulation in Section~\ref{sec:model1}, allowing for general messaging strategies that  use the full information $C_t$. It will be convenient to  write  $C_t$ as $(X_t, \tilde{C}_t)$ where $\tilde{C}_t = \{X_{1:t-1}, U_{1:t-1}^{1,2}, M_{1:t-1}^{1,2}\}$ represents the history of past states, actions and messages. Similarly, every realization $c_t $ of $C_t$ can be written  as the pair $c_t = (x_t, \tilde{c}_t)$. When the specific realization $\tilde{c}_t$ is immaterial, we suppress it for brevity and write $c_t = (x_t, \cdot)$.

We will follow the same approach as in Section \ref{sec:SA-Markovian}. The key idea is that one can view $C_t = (X_t, \tilde{C}_t)$ as the ``current state" that evolves according to the dynamics
\begin{align}\label{eq:dyn_ct}
X_{t+1} &= f_t(X_t, U_t^1, U_t^2, N_t)  \\
    \tilde{C}_{t+1} &= (\tilde{C}_t, X_t, U_t^{1,2}, M^{1,2}_t).
\end{align}
By treating $C_t =(X_t, \tilde{C}_t)$ as an expanded state, a general designer strategy is a ``Markovian" strategy with respect to the expanded state. Therefore, we can now  follow the analysis of Section \ref{sec:SA-Markovian}  with this new expanded state. Following the development in Section \ref{sec:SA-Markovian}, we define value functions $W^i_t(x_t,\tilde{c}_t)$ and $V_t(x_t,\tilde{c}_t)$ as follows:  $V_{T+1}(\cdot) = W^i_{T+1}(\cdot) =0$, and for $t \le T$,
\begin{align}\label{eq:W_def expanded}
    W_t^i(x_t, \tilde{c}_t) &:= \sum_{\substack{m^1 \in \mathcal U_t^1\\ m^2 \in \mathcal U_t^2, n \in \mathcal{N}_t}}
Q_t(n)g^m_t(m^1, m^2 | x_t, \tilde c_t) \notag \\
&\times\left[r_t^i(x_t, m^1, m^2)+ W_{t+1}^i(x_{t+1}, \tilde{c}_{t+1})\right],
\end{align}
\begin{align}\label{eq:V_def expanded}
    V_t(x_t, \tilde{c}_t) &:= \sum_{\substack{m^1 \in \mathcal U_t^1\\ m^2 \in \mathcal U_t^2, n \in \mathcal{N}_t}}
Q_t(n)g^m_t(m^1, m^2 | x_t, \tilde c_t) \notag \\
&\times\left[r_t^0(x_t, m^1, m^2)+ V_{t+1}(x_{t+1}, \tilde{c}_{t+1})\right],
\end{align}
where $x_{t+1} = f_t(x_t, m^{1,2}, n),$ $\tilde{c}_{t+1} =(\tilde{c}_t, x_t,m^{1,2}, m^{1,2})$.
Using arguments analogous to those in Section~\ref{subsec:constraints markov}, we derive a similar result as Lemma~\ref{lemma:sr equiv no comp linear}, which allows the $SR$ constraint in Problem \ref{prob:no comp} to be reformulated as the following set of inequalities  \\
~\\
$\forall {u}\in \mathcal{U}^i_t$,
\begin{align}
    &\sum_{m^j \in \mathcal U_t^j, n \in \mathcal{N}_t}
    Q_t(n)g^m_t(m^j, m_t^i|x_t, \tilde c_t)\notag\\
    &\hspace{20pt}\times\Big[r_t^i(x_t, m^j, m_t^i)+ W^i_{t+1}(x_{t+1}, \tilde{c}_{t+1})\Big]\notag\\
    &\geq \sum_{m^j \in \mathcal U_t^j, n \in \mathcal{N}_t}
    Q_t(n)g^m_t(m^j, m_t^i|x_t, \tilde c_t)\notag\\
    &\hspace{32pt}\times\Big[r_t^i(x_t, m^j, u)+ W^i_{t+1}(x_{t+1}', \tilde{c}_{t+1}')\Big], \label{eq:lp expanded}
\end{align}
where $x_{t+1} = f_t(x_t, m^{1,2}, n),$ $\tilde{c}_{t+1} =(\tilde{c}_t, x_t,m^{1,2}, m^{1,2}),$ $x_{t+1}' = f_t(x_t, m^j, u, n),$ $\tilde{c}_{t+1}' =(\tilde{c}_t, x_t,m^j, u, m^j, m_t^i)$.

Following the analysis of Section \ref{sec:decomposition markov}, we can write a sequence of linear programs $\mathbf{LP}_t(x_t, \tilde c_t)$ as defined below that characterize the optimal designer strategy.
\begin{align*}
&\mathbf{LP}_t(x_t, \tilde c_t):\hspace{10pt}\max_{\substack{g_t^m(\cdot|x_t, \tilde c_t), \\V_t(x_t, \tilde c_t), W_t^1(x_t, \tilde c_t), W_t^2(x_t, \tilde c_t)}} V_t(x_t, \tilde c_t)\notag \\
&\text{s.t. the inequalities in \eqref{eq:lp expanded} hold  for all $m_t^i,\:i=1,2$,} \notag \\
    &\hspace{15pt}W_t^{1,2}(x_t, \tilde c_t) \text{ satisfies }\eqref{eq:W_def expanded} \text{~(with $W^i_{T+1}(\cdot) =0$)}, \notag\\
    &\hspace{15pt}V_t(x_t, \tilde c_t) \text{ satisfies }\eqref{eq:V_def expanded} \text{~(with $V_{T+1}(\cdot) =0$)}.
\end{align*}

\begin{theorem}\label{thm:two}
    The designer strategy $g^m$ where for each time $t$ and each $x_t, \tilde{c}_t$, $g_t^m(\cdot|x_t, \tilde c_t)$ comes from the solution of the linear program $\mathbf{LP}_t(x_t, \tilde c_t)$ is an optimal solution for Problem \ref{prob:no comp}.
\end{theorem}
\begin{proof}
    The proof follows from an analysis similar to Section \ref{sec:SA-Markovian} with $x_t, \tilde{c}_t$ viewed as the ``current state".
\end{proof}
While Theorem \ref{thm:two} provides a method for finding an optimal designer strategy using a sequence of linear programs, the challenge is that the number of linear programs involves increases exponentially in time since the number of possible $\tilde c_t$ increases exponentially. We now argue that the best Markovian strategy found by Algorithm \ref{alg:1} in Section \ref{sec:decomposition markov} is in fact an optimal strategy for Problem \ref{prob:no comp}. 

\begin{theorem}\label{thm:three}
    Let $g_t^m(\cdot|x_t), V_t(x_t), W_t^1(x_t), W_t^2(x_t) $ be an optimal solution of $\mathbf{LP}_t(x_t)$ in Algorithm \ref{alg:1}. Define 
    \begin{align}
        &\hat{g}^m_t(\cdot|x_t, \tilde{c}_t) = g_t^m(\cdot|x_t), \quad  \hat V_t(x_t, \tilde{c}_t) = V_t(x_t) \\
        &\hat W^i_t(x_t, \tilde{c}_t) = W^i_t(x_t), \quad i=1,2.
    \end{align}
    Then, $(\hat{g}^m_t(\cdot|x_t, \tilde{c}_t), \hat V_t(x_t, \tilde{c}_t), \hat W^{1,2}_t(x_t, \tilde{c}_t) )$ is an optimal solution of $\mathbf{LP}_t(x_t,\tilde{c}_t)$. Consequently, the strategy obtained from Algorithm \ref{alg:1} is an optimal designer strategy in Problem~\ref{prob:no comp}.
\end{theorem}
\begin{proof}
    See Appendix~\ref{appendix:proofofthmthree}.
\end{proof}

\section{Example}
\subsection{Model: Multi-access Broadcast System} 
We consider a multi-access broadcast system consisting of two agents that need to transmit packets over a shared channel of limited capacity \cite{schoute1978decentralized, varaiya1979decentralized}. Each agent has a buffer that can store up to $\beta$ packets. Let $B_t^i \in \{0,1, \ldots, \bufcap\}$ denote the number of packets in agent $i$'s buffer  at the beginning of time $t$. Let $A_t^i\in\{0,1\}$ be the number of new packets arriving at agent $i$'s buffer at time $t$. At any time $t$,  $P(A^i_t =1) = p^i$. 

At each time $t$, agent $i$  chooses to transmit $U^i_t \in \{0,1, \ldots, B^i_t\} $. If the total number of packets transmitted by the two agents is no more than the channel capacity $\budcap$, then the packets are successfully transmitted; otherwise the transmission is unsuccessful. Thus, the dynamics of agent $i$'s buffer can be written as:
\begin{align}\label{mab:dyn}
    B_{t+1}^i = \min\left((B_t^i - U_t^i I_{succ})^+ + A_t^i, \bufcap\right),
\end{align}
where $(\cdot)^+ = \max(0,\cdot)$ and $I_{succ} = \mathds{1}_{\{U_t^1 + U_t^2 \leq \budcap\}}$. 

The model described above can be seen as an instance of the system model in Section \ref{sec:model1} with $X_t = (B_t^1, B_t^2)$ and $N_t = (A^1_t, A^2_t)$. We assume that $X_1, A^1_1, A^2_1, \ldots A^1_T,A^2_T $ are independent random variables. 
At each time $t$, the designer and both agents have access to the same information as in \eqref{eq:info structure} in Section~\ref{sec:model1}. The designer generates a message pair $(M_t^1, M_t^2)$ using a distribution $D_t^m = g_t^m(X_t)$ as in Section~\ref{sec:SA-Markovian}.

\emph{Reward Structure:} The designer's reward function is 
\begin{align}\label{mab:designer reward}
    r^0_t(X_t, U_t^1, U_t^2) &:= 
    \left[\frac{U_t^1 + U_t^2}{\budcap} + \alpha \phi(U_t^1, U_t^2)\right] I_{succ} \notag\\
    &- \kappa \left(1-I_{succ}\right) - \lambda^0\sum_{i=1}^2 I_{bound}^i
\end{align}
where $\alpha, \kappa,\lambda^0 \geq 0$, $I_{bound}^i = \mathds{1}_{\{B_t^i = \bufcap\}}$ and $\phi$ measures resource allocation fairness. We define $\phi$ using Jain's fairness index \cite{jain1984quantitative}; specifically for each possible $u_t^1, u^2_t$,
\begin{equation*}
\phi(u_t^1, u_t^2) = \frac{(u_t^1 + u_t^2)^2}{2[(u_t^1)^2 + (u_t^2)^2]},\quad\text{with}\quad\phi(0, 0)= 1.
\end{equation*}
The term $(U_t^1 + U_t^2)/{\budcap}$ in \eqref{mab:designer reward} represents the channel utilization, i.e.,  the ratio of total transmission to channel capacity. 
Finally, the negative terms in \eqref{mab:designer reward} penalize channel capacity violations and full buffers. 

The reward function for agent $i$ is
\begin{align}
r_t^i(B_t^i, U_t^{1,2})&:= U_t^i I_{succ} - \kappa (1-I_{succ})-\lambda^i I_{bound}^i,
\end{align}
where $\lambda^i \geq 0$. Agent $i$ is rewarded for the number of packets it successfully transmits and penalized for capacity violations and a full buffer.

As in Problem \ref{prob:no comp markov}, the designer seeks a messaging strategy that maximizes its total expected reward subject to  $SR$ conditions. To establish a baseline for comparison, we consider a Constrained Markov Decision Process (CMDP) formulation. In the CMDP formulation, the designer can directly control agents' actions with the constraint that each agent's total expected reward is above a threshold. We formalize this below.  

\emph{CMDP formulation:} At each time $t$, the designer directly generates an  \emph{an action pair} $(U_t^1, U_t^2) \in \mathcal{U}_t^1 \times \mathcal{U}_t^2$ for the two agents by sampling from a probability distribution $D^d_t$ on $\mathcal{U}_t^1 \times \mathcal{U}_t^2$. This distribution is chosen by the designer as a function of $X_t$, i.e.,
\begin{equation}\label{eq:messaging strategy2}
  (U^1_t,U^2_t) \sim D_t^d, \text{~~and~~} D_t^d = g_t^d(X_t).
\end{equation}
We refer to $g^d$ as the designer's action strategy and we denote by $\mathcal{G}^d$  the set of all possible strategies for the designer.  The total expected reward for the designer and the two agents are defined analogously to \eqref{eq:designer reward-to-go} and \eqref{eq:agent reward-to-go} as $J^i(g^d) = \mathbb{E}^{g^d}\left[\sum_{t=1}^T r_t^i(X_t, U_t^1, U_t^2)\right]$, for $i=0,1,2$.

The CMDP objective is to find an optimal designer action strategy $g^d$ that maximizes the designer's total expected reward while ensuring that each agent $i$ achieves a minimum total expected reward threshold $\epsilon^i$:
 \begin{align*}
    &\max_{g^d \in \mathcal{G}^d} J^0(g^d)\hspace{10pt}\text{s.t. }J^i(g^d) \geq \epsilon^i,\:i=1,2.
    \end{align*}
Note that as $\epsilon^i\rightarrow -\infty$ for both $i=1,2$, the constraints become trivially satisfied, reducing the CMDP problem above to a standard MDP solvable via backward induction. We denote this specific case as unconstrained designer (UD).

\subsection{Numerical Results}
\subsubsection{Static Case} Consider $T=1$ and $X_1 = (b^1, b^2)$ with probability $1$ where $0 \leq b^2 < b^1 < \bufcap$ and $b^1 + b^2 \leq \budcap$.  Let $\alpha$ in \eqref{mab:designer reward} satisfies 
\begin{equation}\label{mab:alpharegime}
\alpha > \frac{b^1 - b^2}{\budcap\left[1 - \phi(b^1, b^2)\right]}.
\end{equation}
In this case,  the optimal solution to Problem~\ref{prob:no comp markov} can be found analytically. The optimal messaging strategy is $g_1^m(b^1, b^2 | X_1) = 1$ and the optimal designer total expected reward is $(b^1 + b^2)/\budcap + \alpha \phi(b^1, b^2)$.

In the unconstrained designer (UD) problem, the strategy $g_1^d(b^1, b^2 | X_1) = 1$ is suboptimal because it is outperformed by the strategy $g_1^d(b^2, b^2 | X_1) = 1$ which yields designer's objective of $2b^2/\budcap + \alpha$. 
Under \eqref{mab:alpharegime}, it can be easily verified that $2b^2/\budcap + \alpha$ is strictly greater than $(b^1 + b^2)/\budcap + \alpha \phi(b^1, b^2)$. 
\emph{Example:} Let $\budcap = 3, \bufcap = 3, \kappa = 0.1,  \lambda^i = 0.1,i\in\{0,1,2\}$ and $X_1 = (2,1)$. For $\alpha = 4$ (which satisfies \eqref{mab:alpharegime}), the optimal solution for the UD is achieved at $g_1^d(1,1|2,1) = 1$, whereas Problem~\ref{prob:no comp markov} yields $g_1^m(2,1 | 2,1) = 1$. For $\alpha = 3$ (which does not satisfy \eqref{mab:alpharegime}), the two formulations converge to the same optimal strategy $g_1^m(2,1 | 2,1) = g_1^d(2,1 | 2,1) = 1$.

\subsubsection{Dynamic Case} We now set our time horizon to $T=10$ and assume a uniform initial state distribution.  
We set the parameters as follows: $\bufcap = 3, p_1 = p_2 = 0.8, \kappa = 0.1,  \lambda^i = 0.1,i=\{0,1,2\}, \alpha = 4$.
\begin{itemize}
    \item \textbf{Emergence of Mixed Strategies $(\budcap = 3)$:} The optimal strategy in the state $(2,2)$ for the UD is $g_t^d(1,1|2,2) = 1$ for all $t\in\{1,\ldots,10\}$. Conversely, the incentive-constrained designer employs mixed strategies. For instance, given state $(2,2)$ at $t=1$, the optimal strategy assigns probability $0.98$ to the perfectly fair pair $(1,1)$ , and a combined probability $0.02$ to the asymmetric pairs $\{(1,2), (2,1)\}$. As time progresses, the combined probability of recommending the pairs $\{(1,2), (2,1)\}$ in the state $(2,2)$ monotonically increases to $0.64$. This shows that the designer must send asymmetric, unfair pairs to satisfy the $SR$ constraints.
    \item \textbf{Comparison with CMDP $(\budcap = 3)$:} For the CMDP formulation, we use the agents' total expected rewards obtained by the designer's messaging strategy in Problem \ref{prob:no comp markov} as the lower-bound thresholds $\epsilon^i$ in the CMDP. For the UD formulation, we solve the CMDP formulation with $\epsilon^i = -10^{6}$.
    \begin{table}[ht]
    \centering
    \caption{Comparison of UD, CMDP and Problem~\ref{prob:no comp markov}}
    \label{table:reward_comparison}
        \begin{threeparttable}
            \begin{tabular}{|l c c c|}
            \hline
            \textbf{Metric} & \textbf{UD} & \textbf{CMDP} & \textbf{Problem~\ref{prob:no comp markov}} \\ \hline
            Exp. Reward (Designer) & 44.8288 & 43.5472 & 43.3264 \\
            Exp. Reward (Agent 1)  & 8.0065  & 8.6653  & 8.6653  \\
            Exp. Reward (Agent 2)  & 8.0082  & 8.6653  & 8.6653  \\ \hline
            \end{tabular}
        \end{threeparttable}
    \end{table}

As Table \ref{table:reward_comparison} shows, the UD solely maximizes its own objective at the agents' expense. The CMDP formulation yields a slightly higher total expected reward for the designer than in our Problem~\ref{prob:no comp markov}. This happens since the constraints in the CMDP formulation only concern agents' total expected rewards, while the $SR$ constraints must be satisfied for each time and each information realization. This performance gap quantifies the cost of sequential rationality.

 \item \textbf{Over-Capacity Recommendations $(\budcap = 1)$:} We reveal a counter-intuitive phenomenon by reducing the channel capacity to $1$. At $t=7$ in state $(2,2)$, the UD recommends $(0,0)$ to strictly respect the channel capacity while maximizing the fairness metric. In stark contrast, the incentive-constrained designer recommends $(1,1)$, deliberately inducing a capacity violation. This occurs because $(0,0)$ violates the $SR$ conditions. Specifically, when evaluating the $SR$ constraints under the strategy $g_7^m(0,0 | 2,2) = 1$, each agent has a profitable unilateral deviation to action $1$. 
\end{itemize}

\section{Conclusion}\label{sec:conc}
We considered a dynamic action recommendation problem 
 where the designer seeks to coordinate the agents to follow obedient strategies. Assuming complete symmetric information of the designer and the agents, we provided an algorithm for finding an optimal action recommendation/messaging strategy for the designer that optimizes their objective while guaranteeing that obedient strategies are sequential rational for the agents. Our algorithm requires solving a sequence of nested linear programs via backward induction.
 Future work will explore settings with asymmetric information by incentivizing agents through information design and dynamic transfers/compensations.

\appendices
\section{Proof of Lemma \ref{lemma:sr equiv no comp}}\label{app:sr equiv no comp}
\emph{Sufficiency:} We  first show  that for agent $i$ under the condition described in Lemma \ref{lemma:sr equiv no comp}, we can establish \eqref{eq:sr} of Definition \ref{def:sequential rationality}. We first consider time $t=T$ and any $c_T = (x_T,\cdot)\in \mathcal{C}_T$. In this case, the right hand side of \eqref{eq:sr} can be written as (we omit the superscript $(g^m, g^i, \obe^{j})_T$ in some of the expectations below for convenience)
    \begin{align}
    &\mathbb{E}^{(g^m, g^i, \obe^j)_T}\left[ r_T^i(x_T, U_T^{j,i}) | c_T\right] \notag \\
    &= \mathbb{E}\Big[ \mathbb{E}[r_T^i(x_T, U_T^{j,i})|c_T, M_T^i] \Bigm| c_T\Big] \notag \\
    &=\sum_{m_T^i} \Big[ g^m_T(m_T^i|x_T)\times \mathbb{E}^{\mu_T^i(\cdot|m^i_T, x_T)}[r_T^i(x_T, M_T^j, u^i_T)] \Big], \notag \\
    &\le\sum_{m_T^i} \Big[ g^m_T(m_T^i|x_T)\times \mathbb{E}^{\mu_T^i(\cdot|m^i_T, x_T)}[r_T^i(x_T,M_T^j, m_T^i)] \Big],\label{eq:lemma1eq1} 
\end{align}
where $u^i_T = g^i_T(m_T^i,c_T)$. The second equality holds due to conditional independence. Since $c_T$ and $m_T^i$ are fixed in the inner conditioning, the full-information-dependent action $u_T^i$ acts as a deterministic constant. Furthermore, because the designer's messaging strategy $g_T^m$ is Markovian, $M_T^j$ is conditionally independent of the history given $x_T$ and $m_T^i$. Therefore, $\mathbb{P}(M_T^j = m_T^j | m_T^i, c_T) = \mu_T^j(m_T^j | m_T^i, x_T)$ ($\mu_T^j(m_T^j | m_T^i, x_T)$ can be obtained by summing $\mu_T^j(m_T^j , n_T| m_T^i, x_T)$ over all possible $n_T$), allowing us to evaluate the inner expectation exactly using $\mu_T^i(m_T^j | m_T^i, x_T)$. The final inequality follows directly from \eqref{eq:sr equiv no comp}. Repeating the above steps with $\obe^i$ instead of $g^i$ will result in 
\begin{align}
    &\mathbb{E}^{(g^m, \obe^i, \obe^j)_T}\left[ r_T^i(x_T, U_T^{j, i}) | c_T\right] \notag \\
        &=\sum_{m_T^i} \Big[ g^m_T(m_T^i|x_T) 
         \times \mathbb{E}^{\mu_T^i(\cdot|m^i_T, x_T)}[r_T^i(x_T,M_T^j, m_T^i)] \Big]  \label{eq:lemma1eq2} 
\end{align}
\eqref{eq:lemma1eq1} and \eqref{eq:lemma1eq2} establish \eqref{eq:sr} of Definition \ref{def:sequential rationality} for time $T$. Further, using the definition of $\mu_T^i$ from \eqref{eq:mu}, the final expression of right hand side in \eqref{eq:lemma1eq2} can be written as
\begin{align}\label{eq:app1eq8 T}
    &\sum_{m_T^i, m_T^j} g_T^m(m_T^j, m_T^i | x_T)r_T^i(x_T,m_T^j, m_T^i) = W_T^i(x_T)
\end{align}
The expression in \eqref{eq:app1eq8 T} is identical to the definition of $W^i_t(x_t)$ in \eqref{eq:agent val func} when $t=T$. We can now proceed inductively. Assume that (i) \eqref{eq:sr} holds for $t+1$ and each $c_{t+1}$, and (ii) $W_{t+1}^i(x_{t+1})$ is equal to the left hand side of \eqref{eq:sr} for each $c_{t+1}=(x_{t+1} , \cdot)$.
 Consider time $t$ and any $c_t = (x_{t} , \cdot) \in \mathcal{C}_t$. In this case, the right hand side of \eqref{eq:sr} can be written as 
 (we omit the superscript $(g^m, g^i, \obe^j)_{t:T}$ in the final step for convenience)
  \begin{align}
    &\mathbb{E}^{(g^m, g^i, \obe^j)_{t:T}}\Bigg[ \sum_{k=t}^T r_k^i(X_k, U_k^{j,i}) \Bigm| c_t\Bigg] \notag \\
    &= \mathbb{E}\left[r_t^i(x_t, U_t^{j,i}) \Bigm| c_t \right]\notag\\
    &+ \mathbb{E}\left[\mathbb{E}^{(g^m, g^i, \obe^j)_{t+1:T}}\left[ \sum_{k=t+1}^T r_k^i(X_k, U_k^{j,i})\Bigm| C_{t+1}\right]  \Bigm| c_t\right] \notag \\
    &\le \mathbb{E}\Big[r_t^i(x_t, U_t^{j,i})  + W^i_{t+1}(f_t(x_t, U_t^{j,i}, N_t)) \Bigm| c_t\Big] \label{eq:lemma1eq3}
\end{align}
where we used the induction hypothesis in the inequality above. For notational convenience, let 
\begin{align}
    &F^i(x_t, u^{j,i}_t, n_t) = r_t^i(x_t, u_t^{j, i}) + W^i_{t+1}(f_t(x_t, u_t^{j, i}, n_t)).
\end{align}
Then, the right hand side of \eqref{eq:lemma1eq3} can be written as 
\begin{align}
   & \mathbb{E}^{(g^m, g^i, \obe^j)_{t:T}}\left[ F^i(x_t, U^{j,i}_t, N_t)  | c_t\right] \notag \\
    &=\mathbb{E}\Big[ \mathbb{E}[F^i(x_t, M^j_t,g^i_t(M^i_t,c_t),N_t) |c_t,M_t^i] \Bigm| c_t\Big] \notag \\
    &=\sum_{m_t^i} \Big[ g^m_t(m_t^i|x_t)\times \mathbb{E}^{\mu_t^i(\cdot|m^i_t,x_t)}[F^i(x_t, M^j_t, u^i_t,N_t) ] \Big] \label{eq:lemma1eq4} \\
    &\le\sum_{m_t^i} \Big[ g^m_t(m_t^i|x_t)\times \mathbb{E}^{\mu_t^i(\cdot|m^i_t,x_t)}[F^i(x_t, M^j_t, m^i_t,N_t) ] \Big] \label{eq:lemma1eq5}
\end{align}
where $u^i_t = g^i_t(m_t^i,c_t)$. The second equality holds due to conditional independence. Since $c_t$ and $m_t^i$ are fixed in the inner conditioning, the full-information-dependent action $u_t^i$ acts as a deterministic constant. Furthermore, because the designer's messaging strategy $g_t^m$ is Markovian, $M_t^j, N_t$ are conditionally independent of the history given $x_t$ and $m_t^i$. Hence, $\mathbb{P}(M_t^j = m_t^j, N_t = n | m_t^i, c_t) = \mu_t^j(m_t^j, n | m_t^i, x_t)$, allowing us to evaluate the inner expectation exactly using $\mu_t^i(m_t^j, n | m_t^i, x_t)$. The final inequality follows from \eqref{eq:sr equiv no comp} and the definition of $F^i$ in the last inequality. Combining \eqref{eq:lemma1eq3} and \eqref{eq:lemma1eq5}, we obtain
\begin{align}
 &\mathbb{E}^{(g^m, g^i, \obe^j)_{t:T}}\Bigg[ \sum_{k=t}^T r_k^i(X_k, U_k^{j,i}) \Bigm| c_t\Bigg] \notag \\
    &\le\sum_{m_t^i} \Big[ g^m_t(m_t^i|x_t)\times \mathbb{E}^{\mu_t^i(\cdot|m^i_t,x_t)}[F^i(x_t, M^j_t, m^i_t,N_t) ] \Big] \label{eq:lemma1eq6}
\end{align}
Repeating the above steps with $\obe^i$ instead of $g^i$ will result in
\begin{align}
     &\mathbb{E}^{(g^m, \obe^i, \obe^j)_{t:T}}\Bigg[ \sum_{k=t}^T r_k^i(X_k, U_k^{j,i}) \Bigm| c_t\Bigg] \notag \\
     &=\sum_{m_t^i} \Big[ g^m_t(m_t^i|x_t)\times \mathbb{E}^{\mu_t^i(\cdot|m^i_t,x_t)}[F^i(x_t, M_t^j, m_t^i, N_t) ] \Big] \label{eq:lemma1eq7}
     \end{align}
\eqref{eq:lemma1eq6} and \eqref{eq:lemma1eq7} establish \eqref{eq:sr} of Definition \ref{def:sequential rationality} for time $t$. 
Further, using the definition of $\mu_t^i$ from \eqref{eq:mu}, the final  expression in \eqref{eq:lemma1eq7} can be written as
\begin{align}
&\sum_{\substack{m_t^i \in \mathcal U_t^i \\ m_t^j \in \mathcal U_t^j, n \in \mathcal{N}_t}}
Q_t(n)g^m_t(m_t^j, m_t^i | x_t)\notag \\
&\times\left[r_t^i(x_t, m_t^j, m_t^i)+  W_{t+1}^i(f_t(x_t, m_t^j, m_t^i, n))\right].
\label{eq:app1eq8}
\end{align}
The expression in \eqref{eq:app1eq8} is identical to the definition of $W^i_t(x_t)$ in \eqref{eq:agent val func}. This completes the induction argument. The same arguments apply for agent $j$. Hence, \eqref{eq:sr} holds for all $t$.

\emph{Necessity:} We provide a proof-by-contradiction argument. Suppose $(\obe^1,\obe^2)$ satisfies $SR(g^m)$ but \eqref{eq:sr equiv no comp} is not true for some agent $i$ and some time $t$, $1 \le t \le T$, and some realizations $m_t^i, c_t = (x_{t} , \cdot)$ with $g^m_t(m_t^i| x_t) > 0$. Let  ${\tau}$ be the largest time index less than or equal to $T$ such that there exists a $\tilde{m}_{\tau}^i, \tilde{c}_{\tau} = (\tilde{x}_{\tau} , \cdot) $ with $g^m_{\tau}(\tilde{m}_{\tau}^i| \tilde{x}_{\tau}) > 0$  where \eqref{eq:sr equiv no comp} is not true. We will show that each possible value of $\tau$ results in a contradiction.

 Suppose that  $\tau = T$. In this case, we will construct an agent strategy $g^i$ such that \eqref{eq:sr} is violated which will contradict the fact that $(\obe^1,\obe^2)$ satisfies $SR(g^m)$.  Consider the agent action strategy $g^i$ that is identical to $\obe^i$ everywhere except for  realization $\tilde{m}_{T}^i$ and $\tilde{c}_{T}=(\tilde{x}_{T},\cdot)$ of the agent's information at  time $T$.
Define $g_T^i(\tilde{m}_{T}^i,\tilde{c}_{T})$ 
as follows
\begin{align}\label{eq:app2eq1}
&g_T^i(\tilde{m}_{T}^i,\tilde{c}_{T})  \in\argmax_{u \in \mathcal{U}_T^i} \mathbb{E}^{\mu^i_T(\cdot| \tilde{m}_{T}^i, \tilde{x}_T)}[r_T^i(\tilde{x}_T,M_T^j, u)]
\end{align}
For the $g^i$ defined above and the information realization $\tilde{c}_T$, the right hand side of \eqref{eq:sr} can be written as 
\begin{align}
    \mathbb{E}&^{(g^m, g^i, \obe^j)_T}\left[ r_T^i(\tilde{x}_T, U_T^{j,i}) | \tilde{c}_T\right] \notag \\
    &=\mathbb{E}\Big[ \mathbb{E}[r_T^i(\tilde{x}_T, M^j_T, g^i_T(M_T^i,\tilde{c}_T)) |\tilde{c}_T,M_T^i] \Bigm| \tilde{c}_T\Big] \notag \\
    &=\sum_{m_T^i} \Big[ g^m_T(m_T^i|\tilde{x}_T)\times \mathbb{E}^{\mu_T^i(\cdot|m^i_T, \tilde{x}_T)}[r_T^i(\tilde{x}_T, M_T^j, \tilde{u}^i_T)] \Big] \notag\\
    &>\sum_{m_T^i} \Big[ g^m_T(m_T^i|\tilde{x}_T)\times \mathbb{E}^{\mu_T^i(\cdot|m^i_T, \tilde{x}_T)}[r_T^i(\tilde{x}_T,M_T^j, m_T^i)] \Big] \label{eq:app2.a} \\
    & = \mathbb{E}^{(g^m, \obe^i, \obe^j)_T}\left[ r_T^i(\tilde{x}_T, U_T^{j,i})|\tilde{c}_T\right]\label{eq:app2.b}
    \end{align}
where $\tilde{u}^i_T = g^i_T(m_T^i,\tilde{c}_T)$. The strict inequality in \eqref{eq:app2.a} holds because $g^i$ and $\obe^i$ are identical everywhere except at the realization $\tilde{c}_{T}, \tilde{m}_{T}^i$ of the agent's information at  time $T$, and for this critical realization we have 
\begin{align}\label{eq:app2eq3}
    &g^m_T(\tilde{m}_T^i|\tilde{x}_T)\times 
        \mathbb{E}^{\mu_T^i(\cdot| \tilde{m}_{T}^i,\tilde{x}_T)}[r_T^i(\tilde{x}_T, M^j_T, \tilde{u}^i_T)]\notag \\
        &\hspace{10pt}> g^m_T(\tilde{m}_T^i|\tilde{x}_T)\times 
        \mathbb{E}^{\mu_T^i(\cdot| \tilde{m}_{T}^i,\tilde{x}_T)}[r_T^i(\tilde{x}_T, M^j_T, \tilde{m}_T^i)]
\end{align}
since $g^m_T(\tilde{m}_T^i|\tilde{x}_T) >0$ and $\tilde m_T^i$ is not an $\argmax$ of the right hand side of \eqref{eq:sr equiv no comp} for time ${T}$ and the given realization $\tilde{c}_{T}, \tilde{m}_{T}^i$. 

Thus, \eqref{eq:app2.b} shows that the strategy $g^i$ constructed above violates \eqref{eq:sr} which contradicts the fact that $\obe^i$ satisfies sequential rationality. Thus, we must have that $\tau \ne T$.

We can now consider other possible values of $\tau$. Suppose that $\tau =l$, where $1 \le l <T$ and we have $\tilde{m}_{l}^i, \tilde{c}_{l} = (\tilde{x}_{l}, \cdot)$ with $g^m_{l}(\tilde{m}_{l}^i| \tilde{x}_{l}) > 0$  where \eqref{eq:sr equiv no comp} is not true. Consider an action strategy $g^i$ that is identical to $\obe^i$ everywhere except for  realization $\tilde{m}_{l}^i$ and $\tilde{c}_{l} = (\tilde{x}_{l}, \cdot)$ of the agent's information at  time $l$. Define $g_l^i(\tilde{m}_{l}^i,\tilde{c}_{l})$ 
as follows
\begin{align}
     &g_l^i(\tilde{m}_{l}^i,\tilde{c}_{l}) \in \argmax_{u \in \mathcal{U}_l^i} \mathbb{E}^{\mu_l^i(\cdot|\tilde{m}_{l}^i,\tilde{x}_{l})}[r_l^i(\tilde{x}_{l}, M_l^j, u) \notag \\
     &\hspace{40pt}+ W^i_{l+1}(f_l(\tilde{x}_{l}, M_l^j, u, N_t))]
\end{align}
For the $g^i$ defined above and the information realization $\tilde{c}_l$, we can follow the steps used in \eqref{eq:lemma1eq3} to write
\begin{align}
    &\mathbb{E}^{(g^m, g^i, \obe^j)_{l:T}}\Bigg[ \sum_{k=l}^T r_k^i(X_k,U_k^{j,i}) \Bigm| \tilde{c}_l\Bigg] \notag \\
    &= \mathbb{E}\left[r_l^i(\tilde{x}_l, U_l^{j,i})| \tilde{c}_l \right]\notag \\
    &+ \mathbb{E}\Bigg[\mathbb{E}^{(g^m, g^i, \obe^j)_{l+1:T}}\Bigg[ \sum_{k=l+1}^T r_k^i(X_k, U_k^{j,i})\Bigm| C_{l+1}\Bigg]  \Bigm| \tilde{c}_l\Bigg] \notag \\
    &= \mathbb{E}\left[r_l^i(\tilde{x}_l, U_l^{j,i}) | \tilde{c}_l \right]\notag \\
    & + \mathbb{E}\Bigg[ \mathbb{E}^{(g^m, \obe^i, \obe^j)_{l+1:T}}\Bigg[ \sum_{k=l+1}^T r_k^i(X_k, U_k^{j,i}) \Bigm| C_{l+1}\Bigg]  \Bigm| \tilde{c}_l\Bigg] \label{eq:app2.d} \\
    &= \mathbb{E}^{(g^m,g^i, \obe^j)_{l:T}}\Big[r_l^i(\tilde{x}_l, U_l^{j, i})\notag\\
    &\hspace{80pt}+ W^i_{l+1}(f_l(\tilde{x}_l, M_l^j, U_l^i, N_l) \Bigm| \tilde{c}_l\Big] \label{eq:app2.c}
\end{align}
where we used the fact that $g^i$ and $\obe^i$ are identical for time $k > l$ in \eqref{eq:app2.d} and the fact proved in sufficiency part that $W^i_{l+1}(x_{l+1})$ is the left hand side of \eqref{eq:sr} for each $c_{l+1} = (x_{l+1}, \cdot)$.
Now, following the steps  used in sufficiency part to get  \eqref{eq:lemma1eq4} from \eqref{eq:lemma1eq3}, \eqref{eq:app2.c} can be written as
\begin{align}\label{eq:app2eq5}
        &\sum_{m_l^i} \Big[ g^m_l(m_l^i|\tilde{x}_l)\times \mathbb{E}^{\mu_l^i(\cdot| m_{l}^i, \tilde{x}_l)}[F^i(\tilde{x}_l,M_l^{j}, g^i_l(m_l^i,\tilde{c}_l), N_l)] \Big]
\end{align}
Using \eqref{eq:app2eq5}  and  arguments similar to those  in \eqref{eq:app2.a} and \eqref{eq:app2.b}, we obtain
\begin{align}
  &\mathbb{E}^{(g^m, g^i, \obe^j)_{l:T}}\Bigg[ \sum_{k=l}^T r_k^i(X_k, U_k^{j,i}) \Bigm| \tilde{c}_l\Bigg]  \notag \\
  &=\sum_{m_l^i} \Big[ g^m_l(m_l^i|\tilde{x}_l)\times \mathbb{E}^{\mu_l^i(\cdot| m_{l}^i,\tilde{x}_l)}[F^i(\tilde{x}_l, M_l^{j}, \tilde{u}^i_l,N_l)] \Big] \notag \\
  &>\sum_{m_l^i, b^i_l} \Big[ g^m_l(m_l^i|\tilde{x}_l)\times \mathbb{E}^{\mu_l^i(\cdot| m_{l}^i,\tilde{x}_l)}[F^i(\tilde{x}_l, M_l^{j}, m_l^i, N_l)] \Big] \notag \\
  &=\mathbb{E}^{(g^m, \obe^i, \obe^j)_{l:T}}\left[ \sum_{k=l}^T r_k^i(X_k, U_k^{j,i}) \Bigm| \tilde{c}_l\right], \label{eq:app2.e}
\end{align}
where $\tilde{u}^i_l = g^i_l(m_l^i,\tilde{c}_l)$.
\eqref{eq:app2.e} shows that the strategy $g^i$ constructed above violates \eqref{eq:sr} which contradicts the fact that $\obe^i$ satisfies sequential rationality. Thus, we must have that $\tau \ne l$. 

The above argument shows that $\tau$ cannot take any value in $\{T,T-1,\ldots,1\}$. Therefore, \eqref{eq:sr equiv no comp} must hold for each agent $i$, each $t$ and each $c_t, m_t^i$ with $g^m_t(m_t^i| x_t) > 0$.

\section{Proof of Theorem \ref{thm:alg1}}\label{appendix: proofofalgorithm}
We will show that the strategy obtained from Algorithm \ref{alg:1} is an optimal solution to the Global Problem. 

Let $g^{m}_{1:T}, V_{1:T}, W_{1:T}^{1,2}$ be obtained using the sequence of linear programs in Algorithm \ref{alg:1}. That is, for each $t$ and each $x_t$, $(g^{m}_t(\cdot|x_t),W_t^{1,2}(x_t), V_t(x_t))$ is an optimal solution for $\mathbf{LP}_t(x_t)$. It is straightforward to verify that the  $g^{m}_{1:T},  W_{1:T}^{1,2}$ obtained from Algorithm \ref{alg:1} form a feasible solution of the Global Problem since they satisfy all its constraints.  

Let $(g^{m,global}, W^{1,global}_{1:T}, W^{2,global}_{1:T})$ be any  feasible solution for the Global Problem.
Let $g^{global}$ denote the strategy profile $(g^{m,global}, \obe^1, \obe^2)$ and define the reward-to-go functions for the designer under the strategy profile $g^{global}$ as
\allowdisplaybreaks
    \begin{align*}
        V^{global}_t(x_t) := &          \mathbb{E}^{g^{global}_{t:T}}\Bigg[\sum_{k=t}^T r_k^0(X_k, U_k^{1,2}) \Big| X_t = x_t\Bigg].
    \end{align*}
    Note that the objective value of the Global Problem under $(g^{m,global}, W^{1,global}_{1:T}, W^{2,global}_{1:T})$ can be written as
   \begin{equation}
      J^0(g^{m,global},\obe^1,\obe^2) =  \mathbb{E}[V^{global}_1(X_1)].
   \end{equation}  
    We want to show that for all $t=T, T-1, \ldots,1$ and for each $x_t \in \mathcal{X}_t$, we have $V^{global}_t(x_t) \leq V_t(x_t)$. In other words, the value functions $V_t(\cdot)$ obtained from Algorithm \ref{alg:1} dominate the designer's reward-to-go functions $V^{global}_t(\cdot)$ for any feasible solution of the Global Problem.
    
    \emph{Base case  $(t = T)$:} Fix a $x_T \in \mathcal{X}_T$.  Then
    \begin{align}
         &V^{global}_T(x_T) = \mathbb{E}^{g^{global}_T}\left[r_T^0(X_T, U_T^{1,2}) | X_T = x_T\right] \notag \\
         &=\sum_{\substack{m^1 \in \mathcal U_T^1, m^2 \in \mathcal U_T^2}}
g^{m,global}_T(m^{1, 2} | x_T) \times r_T^0(x_T,m^{1,2})\notag\\
         &=\sum_{\substack{m^1 \in \mathcal U_T^1\\ m^2 \in \mathcal U_T^2, n \in \mathcal{N}_T}}
Q_T(n)g^{m,global}_T(m^{1, 2} | x_T)\times r_T^0(x_T,m^{1,2})
    \end{align}
    
    It is now easy to verify that the tuple $(g^{m,global}_T(\cdot|x_T),$ $ W^{1,global}_T(x_T), W^{2,global}_T(x_T), V^{global}_T(x_T))$ is a feasible solution for $\mathbf{LP}_T(x_T)$. Thus, it follows that
    \begin{equation}
        V^{global}_T(x_T) \le V_T(x_T) \label{eq:app3.b}
    \end{equation} 
    since $V_T(x_T)$ comes from the optimal solution for $\mathbf{LP}_T(x_T)$.

    \emph{Induction step:} Now suppose that  $V^{global}_t(\cdot) \leq V_{t}(\cdot)$ holds for all $t \ge l+1$.  Fix a $c_l = (x_l,\cdot)$, we obtain
    \begin{align}
        &V^{global}_l(x_l) :=           \mathbb{E}^{g^{global}_{l:T}}\left[\sum_{k=l}^T r_k^0(X_k, U_k^{1,2}) \Big| X_l = x_l\right] \notag\\
        &= \mathbb{E}^{g^{global}_{l:T}}\Big[r_l^0(x_l, U_l^{1,2}) + V^{global}_{l+1}(X_{l+1}) \Big| X_l = x_l\Big] \notag \\
        &=\sum_{\substack{m^1 \in \mathcal U_l^1\\  m^2 \in \mathcal U_l^2, n \in \mathcal{N}_l}}
 Q_l(n)g_l^{m,global}(m^1, m^2 | x_l) \notag \\
        &\hspace{10pt} \times\left[r_l^0(x_l, m^1, m^2) + V^{global}_{l+1}(f_l(x_l,m^{1,2}, n))\right].\label{eq:app3.a}
            \end{align}
     By the induction hypothesis, $V^{global}_{l+1}(\cdot) \le V_{l+1}(\cdot)$. Hence  the expression in \eqref{eq:app3.a} is upper bounded as follows    
\begin{align}
 \eqref{eq:app3.a}   &\le \sum_{\substack{m^1 \in \mathcal U_l^1\\ m^2 \in \mathcal U_l^2, n \in \mathcal{N}_l}}
 Q_l(n)g_l^{m,global}(m^1, m^2 | x_l) \notag \\
        &\hspace{8pt} \times\left[r_l^0(x_l, m^1, m^2) + V_{l+1}(f_t(x_l,m^1, m^2, n))\right]  \label{eq:app3eq1} \\
&:= \hat{V}_l(x_l).
\end{align} 
It is now easy to verify that the tuple $(g^{m,global}_l(\cdot|x_l),$ $ W^{1,global}_l(x_l), W^{2,global}_l(x_l), \hat{V}_l(x_l))$ is a feasible solution for $\mathbf{LP}_l(x_l)$. Thus, it follows that $\hat{V}_l(x_l) \le V_l(x_l)$ (since $V_l(x_l)$ comes from the optimal solution for $\mathbf{LP}_l(x_l)$).  Combining this with \eqref{eq:app3eq1}, we get $V^{global}_l(x_l) \le \hat{V}_l(x_l) \le V_l(x_l)$. This completes the induction argument.

Hence,  at time $1$, we have that $V^{global}_{1}(\cdot) \leq V_{1}(\cdot)$.
    Therefore, the objective value of the Global Problem under $(g^{m,global}, W^{1,global}_{1:T}, W^{2,global}_{1:T})$, which is equal to $\mathbb{E}[V^{global}_1(X_1)] $, satisfies 
    \begin{align}
     J^0(g^{m,global},\obe^1,\obe^2) =    \mathbb{E}[V^{global}_1(X_1)] \le \mathbb{E}[V_1(X_1)]. \label{eq:app3eq2}
    \end{align} 
    Repeating the above arguments with $g^{m}_{1:T}, W_{1:T}^{1,2}$ obtained from Algorithm \ref{alg:1} instead of  $(g^{m,global}, W^{1,global}_{1:T}, W^{2,global}_{1:T})$ will change all inequalities to equalities. 
    In particular, we will get
    \begin{align}
     J^0(g^{m},\obe^1,\obe^2) =    \mathbb{E}[V_1(X_1)]. \label{eq:app3eq2a}
    \end{align} 
 Comparing \eqref{eq:app3eq2} and \eqref{eq:app3eq2a}, it is clear that the Markovian strategy $g^{m}$ obtained from Algorithm \ref{alg:1} is optimal for the Global Problem and hence for Problem~\ref{prob:no comp markov}. 

 \section{Proof of Theorem~\ref{thm:three}}\label{appendix:proofofthmthree}
    First consider time $t=T$ and the linear programs \( \mathbf{LP}_T(x_T,\tilde{c}_T) \) and \( \mathbf{LP}_T(x_T) \). The terms multiplying $g^m_T(\cdot)$ in \eqref{eq:lp expanded}, \eqref{eq:W_def expanded}, \eqref{eq:V_def expanded}  are identical to those in \eqref{eq:lp3}, \eqref{eq:agent val func}, \eqref{eq:designer val func}. Thus, the  two linear programs are identical and share an optimal solution. Consequently, the optimal  value functions can be sufficiently expressed by the state $x_T$ alone. This establishes that $(\hat{g}^m_T(\cdot|x_T, \tilde{c}_T), \hat V_T(x_T, \tilde{c}_T), \hat W^{1,2}_T(x_T, \tilde{c}_T) )$ optimally solves $\mathbf{LP}_T(x_T,\tilde{c}_T)$. 
    We can now proceed inductively. Assume that $(\hat{g}^m_{t+1}(\cdot|x_{t+1}, \tilde{c}_{t+1}), \hat V_{t+1}(x_{t+1}, \tilde{c}_{t+1}), \hat W^{1,2}_{t+1}(x_{t+1}, \tilde{c}_{t+1}) )$ is an optimal solution of $\mathbf{LP}_{t+1}(x_{t+1},\tilde{c}_{t+1})$.   We can proceed as we did at time $T$: the terms multiplying $g^m_t(\cdot)$ in \eqref{eq:lp expanded}, \eqref{eq:W_def expanded}, \eqref{eq:V_def expanded}  are identical to those in \eqref{eq:lp3}, \eqref{eq:agent val func}, \eqref{eq:designer val func}. Thus, the linear programs \( \mathbf{LP}_t(x_t,\tilde{c}_t) \) and \( \mathbf{LP}_t(x_t) \) share an optimal solution and the optimal  value functions can be sufficiently expressed by the state $x_t$. This demonstrates that  $(\hat{g}^m_t(\cdot|x_t, \tilde{c}_t), \hat V_t(x_t, \tilde{c}_t), \hat W^{1,2}_t(x_t, \tilde{c}_t) )$ is an optimal solution of $\mathbf{LP}_t(x_t,\tilde{c}_t)$, completing the induction argument. 
\bibliography{ref_journal_2025.bib}
\end{document}